\documentclass[final]{elsarticle}

\usepackage{amsfonts}
\usepackage{color}
\usepackage{graphicx}
\usepackage[dvips]{epsfig}
\usepackage{graphics} 
\usepackage{times} 
\usepackage[cmex10]{amsmath} 
\usepackage{amssymb}  
\usepackage{multirow}
\usepackage[tight,footnotesize]{subfigure}
\usepackage{amsmath}
\usepackage[boxed,ruled,lined]{algorithm2e}
\usepackage{mathtools}
\usepackage{amsthm}

\def\norm #1{\left\|#1\right\|}

\def\twon #1{\left\|#1\right\|_2}
\def\onen #1{\left\|#1\right\|_1}

\def\frobn #1{\left\|#1\right\|_{\text{F}}}

\def\sgn #1{\text{sgn}#1}
\def\abs #1{\left|#1\right|}

\def\st{\text{subject to }}

\def\bC{\mathbb{C}}

\def\bR{\mathbb{R}}

\def\bT{\mathbb{T}}

\def\m #1{\boldsymbol{#1}}

\def\cA{\mathcal{A}}

\def\cI{\mathcal{I}}

\def\cK{\mathcal{K}}
\def\cL{\mathcal{L}}

\def\bee{\begin{equation}}
\def\ene{\end{equation}}

\def\beq{\begin{eqnarray}}
\def\enq{\end{eqnarray}}
\def\lentwo{\setlength\arraycolsep{2pt}}

\newtheorem{lem}{Lemma}
\newtheorem{rem}{Remark}
\newtheorem{cor}{Corollary}
\newtheorem{thm}{Theorem}

\newtheorem{exa}{Example}

\def\equ #1{\begin{equation}#1\end{equation}}
\def\equa #1{\begin{eqnarray}#1\end{eqnarray}}
\def\sbra #1{\left(#1\right)}
\def\mbra #1{\left[#1\right]}
\def\lbra #1{\left\{#1\right\}}
\def\diag #1{\text{diag}#1}
\def\tr #1{\text{tr}#1}
\def\supp #1{\text{supp}#1}
\def\rank #1{\text{rank}#1}
\def\st {\text{ subject to }}

\begin{document}

\begin{frontmatter}

\title{Frequency-Selective Vandermonde Decomposition of Toeplitz Matrices with Applications}

\author{Zai Yang\fnref{fn1}}
\ead{yangzai@njust.edu.cn}
\author{Lihua Xie\fnref{fn3}}
\ead{elhxie@ntu.edu.sg}
\fntext[fn1]{School of Automation, Nanjing University of Science and Technology, Nanjing 210094, China}
\fntext[fn3]{School of Electrical and Electronic Engineering, Nanyang Technological University, Singapore 639798\newline \phantom{this}This paper has been accepted by {\em Signal Processing} (manuscript in May 2016; last updated in July 2017; first presented at Chinese Control Conference (CCC) in July 2016).}





\begin{abstract} The classical result of Vandermonde decomposition of positive semidefinite Toeplitz matrices, which dates back to the early twentieth century, forms the basis of modern subspace and recent atomic norm methods for frequency estimation. In this paper, we study the Vandermonde decomposition in which the frequencies are restricted to lie in a given interval, referred to as frequency-selective Vandermonde decomposition. The existence and uniqueness of the decomposition are studied under explicit conditions on the Toeplitz matrix. The new result is connected by duality to the positive real lemma for trigonometric polynomials nonnegative on the same frequency interval. Its applications in the theory of moments and line spectral estimation are illustrated. In particular, it provides a solution to the truncated trigonometric $K$-moment problem. It is used to derive a primal semidefinite program formulation of the frequency-selective atomic norm in which the frequencies are known {\em a priori} to lie in certain frequency bands. Numerical examples are also provided.
\end{abstract}


\end{frontmatter}

\section{Introduction}

A classical result discovered by Carath\'{e}odory and Fej\'{e}r in 1911 \cite{caratheodory1911zusammenhang} states that, if an $N\times N$ Hermitian Toeplitz matrix $\m{T}$ is positive semidefinite (PSD) and has rank $r\leq N$, then it can be factorized as
\equ{\m{T} = \m{A}\m{P}\m{A}^H, \label{eq:vanderdecp1}}
where $\m{P}$ is an $r\times r$ positive definite diagonal matrix and $\m{A}$ is an $N\times r$ Vandermonde matrix whose columns are discrete sinusoidal waves with distinct frequencies. Moreover, such a decomposition is unique if $r<N$. This Vandermonde decomposition result has become important for information and signal processing since the 1970s when it was rediscovered by Pisarenko and used for frequency estimation by interpreting the Toeplitz matrix $\m{T}$ as the data covariance matrix. The Vandermonde decomposition in \eqref{eq:vanderdecp1} is therefore also referred to as the Carath\'{e}odory-Fej\'{e}r-Pisarenko decomposition. As a result of this rediscovery, a class of methods have been developed for frequency estimation based on the signal subspace of a data covariance estimate, known as the subspace-based methods. Prominent examples are multiple signal classification (MUSIC), estimation of parameters by rotational invariant techniques (ESPRIT) and various variants of them (see the review in \cite{stoica2005spectral}). Besides, this decomposition result is important in moment theory, operator theory and system theory \cite{grenander1958toeplitz, georgiou2007caratheodory}. As an example, it can be applied to give a solution to the truncated trigonometric moment problem (a.k.a.~the moment problem on the unit circle given a finite moment sequence) \cite{akhiezer1965classical}.

In the past few years, a new class of methods for frequency estimation have been devised, namely the gridless sparse methods (see the review in \cite{yang2016sparse}), in which the Vandermonde decomposition is evoked and plays an important role. It is well-known that sparse methods for frequency estimation developed in the past two decades exploit the signal sparsity, which arises naturally from the fact that the number of frequencies is small, and attempt to find, among all candidates consistent with the observed data, the solution consisting of the smallest number of frequencies. Since frequency estimation is a highly nonlinear problem and to overcome such nonlinearity, gridding in the continuous frequency domain used to be a standard ingredient of early sparse methods, which transforms approximately the original nonlinear continuous parameter estimation problem as a problem of sparse signal recovery from a linear system of equations (see, e.g., \cite{gorodnitsky1997sparse,malioutov2005sparse}). The newly developed gridless sparse methods completely avoid gridding, work directly in the continuous domain, and have strong theoretical guarantees. These methods have been developed based on the atomic norm \cite{chandrasekaran2012convex,candes2013towards,candes2013super, bhaskar2013atomic,tang2012compressed}---a continuous analogue of the $\ell_1$ norm used in the early sparse methods---and covariance fitting \cite{yang2015gridless}. A main difficulty of applying these gridless sparse methods underlies in how to solve the nonlinearity problem, which makes the resulting optimization problems nonconvex with respect to the unknown frequencies. To do so, the key is to apply the Vandermonde decomposition of Toeplitz matrices to cast these optimization problems as semidefinite programs (SDP), in which the frequencies are encoded in a PSD Toeplitz matrix, as $\m{T}$ in \eqref{eq:vanderdecp1}. Once the SDP is solved, the frequencies are finally retrieved from the Vandermonde decomposition of the solved Toeplitz matrix. Note that the Vandermonde decomposition result has also been generalized to high dimensions and used for multidimensional frequency estimation \cite{yang2016vandermonde}.

Notice that the frequencies in the Vandermonde decomposition in \eqref{eq:vanderdecp1} may take any value in the normalized band $\mbra{0,1}$ (or the unit circle), in which 0 and 1 are identified. This paper is motivated by various practical applications in which the (normalized) frequencies can be known {\em a priori} to lie in certain frequency bands. For example, when a signal is oversampled by a factor, the frequencies will lie in a band narrowed by the same factor. Due to the path loss effect, the maximum value of the range/delay, which can be interpreted as a frequency parameter, of a detectable aircraft can be estimated in advance. Similarly, the maximum Doppler frequency can be obtained if the aircraft's characteristic speed can be known. In underwater channel estimation, the frequency parameters of interest can reside in a known small interval \cite{beygi2015multi}. Similar prior knowledge might also be available given weather observations \cite{doviak1993doppler}. Therefore, it would be interesting to exploit such prior knowledge in gridless sparse methods for frequency estimation, and by doing so, the estimation accuracy is expected to improve.

The important role of the Vandermonde decomposition in gridless sparse methods encourages us to incorporate the prior interval knowledge into the decomposition. In other words, we ask the following question: {\em Can the frequencies in the Vandermonde decomposition of the Toeplitz matrix $\m{T}$, as in \eqref{eq:vanderdecp1}, be restricted to lie in a given interval $\cI\subset\mbra{0,1}$, instead of the entire domain $\mbra{0,1}$, under explicit conditions on $\m{T}$?} In fact, we also want the conditions to be convex due to our interest in optimization problems. The resulting decomposition is referred to as frequency-selective (FS) Vandermonde decomposition. The question asked above is challenging since, by \eqref{eq:vanderdecp1}, $\m{T}$ is a highly nonlinear function of the frequencies and it is unclear how to link $\m{T}$ to a frequency interval $\cI$.

It is interesting to note that similar questions have been investigated in a class of moment problems known as truncated $K$-moment problems, a.k.a.~truncated moment problems on a semialgebraic set $K$, instead of on an entire domain \cite{schmudgen1991thek}. When $K$ is in the real or the complex domain, solutions to these problems have been successfully obtained \cite{curto2000truncated,lasserre2009moments}. To the best of our knowledge, however, the problem is still open when $K$ is defined on the unit circle $\mbra{0,1}$, which is known as the truncated trigonometric $K$-moment problem. In this paper, we show that the study of the FS Vandermonde decomposition can provide a solution to this open problem.

In this paper, an affirmative answer is provided to the question asked above. Concretely, it is shown that a PSD Toeplitz matrix $\m{T}$ admits an FS Vandermonde decomposition on a given interval if and only if $\m{T}$ satisfies another linear matrix inequality (LMI). Interestingly, this FS Vandermonde decomposition result is linked by duality to the positive real lemma (PRL) for trigonometric polynomials \cite{dumitrescu2007positive}. The usefulness of the new result is also demonstrated. In the theory of moments, it provides a solution to the truncated trigonometric $K$-moment problem. For frequency estimation with prior interval knowledge, it is used to derive a primal SDP formulation for the atomic norm exploiting the prior knowledge. Numerical examples are also provided.

\subsection{Related Work}

This paper extends our conference paper \cite{yang2016vandermonde_ccc} in which the FS Vandermonde decomposition of Toeplitz matrices was studied. In addition to this, we show in this paper the connection between the FS Vandermonde decomposition and the PRL for trigonometric polynomials. Its applications to the moment theory and frequency estimation are also studied in more detail.

The problem of frequency estimation with restriction on the frequency band was studied in \cite{mishra2015spectral,yang2016weighted, chao2016extensions}. In \cite{mishra2015spectral}, an FS atomic norm formulation (or constrained atomic norm in the language of \cite{mishra2015spectral}) was proposed and a dual SDP formulation was presented by applying the theory of positive trigonometric polynomials. In contrast to this, we show in this paper that a primal SDP formulation of the FS atomic norm can be obtained by applying the new FS Vandermonde decomposition. In \cite{yang2016weighted}, the interval prior was interpreted as a prior distribution of the frequencies and a weighted atomic norm approach was then devised that is an approximate but faster implementation of the FS atomic norm. Although the paper \cite{chao2016extensions} does not provide or imply the FS Vandermonde decomposition result, it obtained independently a primal SDP formulation of the FS atomic norm based on a different technique.

The paper \cite{de2015exact} studied the super-resolution problem on semialgebraic sets in the real domain and provided an SDP formulation of the resulting atomic norm. To do so, the key is to apply the moment theory on semialgebraic sets in the real domain (a.k.a. the truncated $K$-moment problem in the real domain). In contrast to this, we provide a first solution to the truncated trigonometric $K$-moment problem and then apply this result to study super-resolution on semi-algebraic sets on the unit circle.


\subsection{Notations}



Notations used in this paper are as follows. $\bR$ and $\bC$ denote the set of real and complex numbers, respectively. $\bT\coloneqq\mbra{0,1}$ denotes the unit circle, in which 0 and 1 are identified. Boldface letters are reserved for vectors and matrices. $\abs{\cdot}$ denotes the amplitude of a scalar or the cardinality of a set. $\onen{\cdot}$, $\twon{\cdot}$ and $\frobn{\cdot}$ denote the
$\ell_1$, $\ell_2$ and Frobenius norms respectively. $\m{A}^T$ and $\m{A}^H$ are the matrix transpose and conjugate transpose of $\m{A}$ respectively. $\rank\sbra{\m{A}}$ denotes the rank and $\tr\sbra{\m{A}}$ is the trace. For PSD matrices $\m{A}$ and $\m{B}$, $\m{A}\geq\m{B}$ means that $\m{A}-\m{B}$ is PSD. $\Re$ and $\Im$ return the real and the imaginary parts of a complex argument respectively.

A Hermitian trigonometric polynomial of degree one is defined as:
\equ{g(z) = r_1 z^{-1} + r_0 + r_{-1} z, \quad r_{-1}=\overline{r}_1,\quad r_0\in\bR, \label{eq:gz}}
where $z$ is a complex argument and $\overline{\cdot}$ denotes the complex conjugate operator. When $\m{z}$ is on the unit circle, i.e., when $z = e^{i2\pi f}$, $f\in\bT$, we write without ambiguity $g(f) \coloneqq g\sbra{e^{i2\pi f}}$. It follows that
\equ{g(f) = r_1e^{-i2\pi f} + r_0 + \overline{r}_1 e^{i2\pi f} = r_0 + 2\Re\lbra{r_1e^{-i2\pi f}}, }
and $g(f)$ is real on $\bT$.

An $N\times N$ Toeplitz matrix $\m{T}\coloneqq \m{T}\sbra{\m{t}}\coloneqq \m{T}\sbra{N,\m{t}}$ is formed by using a complex sequence $\m{t} = \mbra{t_j}$, $j=1-N,\dots,N-1$ and defined by $T_{mn} = t_{n-m}$, $0\leq m,n\leq N-1$. Given $\m{t}$ and a degree-1 trigonometric polynomial $g$ as defined in \eqref{eq:gz}, an $(N-1)\times (N-1)$ Toeplitz matrix $\m{T}_{g}\coloneqq \m{T}_{g}\sbra{\m{t}}\coloneqq \m{T}_{g}\sbra{N,\m{t}}$ is defined by
\equ{\mbra{T_g}_{mn} = r_1 t_{n-m+1}+ r_0t_{n-m} + r_{-1}t_{n-m-1}, \label{eq:Tg}}
$0\leq m,n\leq N-2$.
Also, let $\m{a}\sbra{f} \coloneqq \m{a}\sbra{N,f}\coloneqq \mbra{1,e^{i2\pi f}, \dots, e^{i2\pi (N-1)f}}^T$ denote a size-$N$ discrete complex sinusoid with frequency $f\in\bT$.

\subsection{Paper Organization}
The rest of the paper is organized as follows. Section \ref{sec:standardVD} introduces the standard Vandermonde decomposition of Toeplitz matrices. Section \ref{sec:VDint} presents the new FS Vandermonde decomposition. Section \ref{sec:duality} shows connections between the new result and the theory of trigonometric polynomials. Section \ref{sec:moment} illustrates its application in the theory of moments. Section \ref{sec:application} turns to the application in line spectral estimation with prior knowledge. Section \ref{sec:conclusion} concludes this paper.

\section{Vandermonde Decomposition of Toeplitz Matrices} \label{sec:standardVD}

The standard Vandermonde decomposition theorem of Toeplitz matrices \cite{caratheodory1911zusammenhang,stoica2005spectral} is summarized in this section. Although its proof can be found in, e.g., \cite{stoica2005spectral}, a new proof, inspired by \cite{gurvits2002largest}, is provided here which will form the basis of the proof of the FS Vandermonde decomposition given in Section \ref{sec:VDint}.

\begin{thm} A Toeplitz matrix $\m{T}\in\bC^{N\times N}$ admits the following $r$-atomic, $r=\rank\sbra{\m{T}}$, Vandermonde decomposition:
\equ{\m{T} = \sum_{k=1}^r p_k \m{a}\sbra{f_k}\m{a}^H\sbra{f_k}, \label{eq:VD}}
where $f_k\in\bT$, $k=1,\dots,r$ are distinct and $p_k>0$, if and only if $\m{T} \geq \m{0}$. Moreover, the decomposition is unique if $\m{T}$ is rank-deficient. \label{thm:VD}
\end{thm}
\begin{proof} Suppose that $\m{T}$ can be written as in \eqref{eq:VD}, where $p_k>0$, it is evident that $\m{T}$ is PSD. This completes the `only if' part. We next show the `if' part. To do so, we start with the case of $r=\rank\sbra{\m{T}}\leq N-1$. Since $\m{T}\geq \m{0}$, there exists $\m{V} = \mbra{\m{v}_1^T,\dots,\m{v}_{N}^T}^T\in\bC^{N\times r}$ satisfying $\m{T} = \m{V}\m{V}^H$, where $\m{v}_j\in\bC^{1\times r}$, $j=1,\dots,N$. Let $\m{V}_U = \mbra{\m{v}_1^T,\dots, \m{v}_{N-1}^T}^T$ and $\m{V}_L = \mbra{\m{v}_2^T,\dots, \m{v}_{N}^T}^T$. By the structure of $\m{T}$, we have that $\m{V}_U\m{V}_U^H = \m{V}_L\m{V}_L^H$. By \cite[Theorem 7.3.11]{horn2012matrix}, there exists an $r\times r$ unitary matrix $\m{U}$ satisfying $\m{V}_L = \m{V}_U\m{U}$. It follows that $\m{v}_j = \m{v}_1\m{U}^{j-1}$, $j=2,\dots,N$ and therefore,
\equ{t_j = \m{v}_1 \m{U}^{-j}\m{v}_1^H, \quad j=1-N,\dots,N-1. \label{eq:tjinvU}}
Note that $\m{U}$ has the following eigen-decomposition:
\equ{\m{U} = \widetilde{\m{U}} \diag\sbra{z_1,\dots,z_r} \widetilde{\m{U}}^H, \label{eqw:UinUz}}
where $\widetilde{\m{U}}$ is also an $r\times r$ unitary matrix and $z_k = e^{i2\pi f_k}$ with $f_k\in\bT$, $k=1,\dots,r$. Insert \eqref{eqw:UinUz} into \eqref{eq:tjinvU} and let $p_k = \abs{\m{v}_1\widetilde{\m{u}}_k}^2>0$, $k=1,\dots,r$, where $\widetilde{\m{u}}_k$ denotes the $k$th column of $\widetilde{\m{U}}$. Then we have that
\equ{t_j = \sum_{k=1}^r p_k e^{-i2\pi j f_k}.}
Using the identity above, $\m{T}$ can be written as in \eqref{eq:VD}. It is evident that $f_k$, $k=1,\dots,r$ are distinct since otherwise, $\rank\sbra{\m{T}} < r$, which cannot be true.

We now consider the case of $r=N$, in which $\m{T}$ is positive definite. To obtain a decomposition as in \eqref{eq:VD}, we choose arbitrarily $f_N \in\bT$ and let $p_N = \sbra{\m{a}^H\sbra{f_N} \m{T}^{-1} \m{a}\sbra{f_N}}^{-1}>0$. After that, we define a new sequence $\m{t}'=\mbra{t'_j}, \;\abs{j}\leq N-1$ as:
\equ{t'_j = t_j - p_N e^{-i2\pi j f_N}. \label{eq:t1j}}
It follows that
\equ{\m{T}\sbra{\m{t}'}
= \m{T} - p_N\m{a}\sbra{f_N}\m{a}^H\sbra{f_N}. \label{eq:Tu1}}
By the choice of $p_N$, the matrix
\equ{\begin{bmatrix} p_N^{-1} & \m{a}^H\sbra{f_N} \\ \m{a}\sbra{f_N} & \m{T} \end{bmatrix} = \begin{bmatrix} \m{a}^H\sbra{f_N}\m{T}^{-\frac{1}{2}} \\ \m{T}^{\frac{1}{2}} \end{bmatrix}\begin{bmatrix} \m{a}^H\sbra{f_N}\m{T}^{-\frac{1}{2}} \\ \m{T}^{\frac{1}{2}} \end{bmatrix}^H \notag}
is PSD and rank-deficient.
Notice that $\m{T}\sbra{\m{t}'}$ is the Schur complement of $\m{T}$ in the above matrix, and therefore
\equ{\m{T}\sbra{\m{t}'}\geq \m{0}. \label{eq:Tu1psd}}
Moreover, it holds that
\equ{\rank\sbra{\m{T}\sbra{\m{t}'}}<N \label{eq:Tu1rankd}}
since, otherwise, $\begin{bmatrix} p_N^{-1} & \m{a}^H\sbra{f_N} \\ \m{a}\sbra{f_N} & \m{T} \end{bmatrix}$ has full rank.
Combining \eqref{eq:Tu1rankd} and
\equ{\rank\sbra{\m{T}\sbra{\m{t}'}} \geq \rank\sbra{\m{T}} - \rank\sbra{p_N\m{a}\sbra{f_N}\m{a}^H\sbra{f_N}} = N-1 \notag}
results in
\equ{\rank\sbra{\m{T}\sbra{\m{t}'}}= N-1.\label{eq:Tu1rank}}
Following from \eqref{eq:Tu1psd}, \eqref{eq:Tu1rank} and the result in the case of $r\leq N-1$ that we just proved, $\m{T}\sbra{\m{t}'}$ admits a Vandermonde decomposition as in \eqref{eq:VD} with $r=N-1$. It then follows from \eqref{eq:Tu1} that $\m{T}$ admits an $N$-atomic Vandermonde decomposition.

We finally show the uniqueness in the case of $r\leq N-1$. Write \eqref{eq:VD} in matrix form as $\m{T}=\m{A}\sbra{\m{f}} \m{P}\m{A}^H\sbra{\m{f}}$, where $\m{P} = \diag\sbra{p_1,\dots,p_r}$ and $\m{A}\sbra{\m{f}} = \mbra{\m{a}\sbra{f_1},\dots, \m{a}\sbra{f_r}}$. Suppose that $\m{T}$ has another decomposition: $\m{T} = \m{A}\sbra{\m{f}'} \m{P}'\m{A}^H\sbra{\m{f}'}$, in which, similarly, $f'_j\in\bT$, $j=1,\dots,r$ are distinct and $p'_j>0$. It is evident that
\equ{\m{A}\sbra{\m{f}'} \m{P}'\m{A}^H\sbra{\m{f}'}=\m{A}\sbra{\m{f}} \m{P}\m{A}^H\sbra{\m{f}}.}
Therefore, there exists an $r\times r$ unitary matrix $\m{U}'$ satifying $\m{A}\sbra{\m{f}'}\m{P}'^{\frac{1}{2}} = \m{A}\sbra{\m{f}} \m{P}^{\frac{1}{2}} \m{U}'$. It follows that
\equ{\m{A}\sbra{\m{f}'} = \m{A}\sbra{\m{f}} \m{P}^{\frac{1}{2}} \m{U}' \m{P}'^{-\frac{1}{2}}.}
This means that for every $j = 1,\dots,r$, $\m{a}\sbra{f'_j}$ lies in the range space spanned by $\lbra{\m{a}\sbra{f_k}}_{k=1}^r$. By the fact that $r\leq N-1$ and that any $N$ atoms $\m{a}\sbra{f_k}$ with distinct $f_k$'s are linearly independent, we have that $f'_j\in\lbra{f_k}_{k=1}^r$ and thus the two sets $\lbra{f'_j}_{j=1}^r$ and $\lbra{f_k}_{k=1}^r$ are identical. It follows that the two decompositions of $\m{T}$ are identical.
\end{proof}

We next discuss how to obtain the Vandermonde decomposition, to be specific, how to solve for $f_k$ and $p_k$ in \eqref{eq:VD}. In fact, a computational approach can be provided based on the proof of Theorem \ref{thm:VD}. In the case of $r\leq N-1$, using Cholesky decomposition, we can compute $\m{V}\in\bC^{N\times r}$ satisfying $\m{T} = \m{V}\m{V}^H$. By the arguments of the proof, it is easy to show the following equation:
\equ{\sbra{\m{V}_U^H\m{V}_L - z_k \m{V}_U^H\m{V}_U}\widetilde{\m{u}}_k = 0,}
from which $z_k$ and $\widetilde{\m{u}}_k$, $k=1,\dots,r$ can be computed as the eigenvalues and eigenvectors of the matrix pencil $\sbra{\m{V}_U^H\m{V}_L, \m{V}_U^H\m{V}_U}$. Finally, the parameters are obtained as: $f_k = \frac{1}{2\pi}\Im \ln z_k\in\bT$ and $p_k = \abs{\m{v}_1\widetilde{\m{u}}_k}^2$, $k=1,\dots,r$, where $\m{v}_1$ is the first row of $\m{V}$. In the case of $r= N$, $f_N\in\bT$ can be chosen arbitrarily first, and the rest can be done following from the proof.

\section{FS Vandermonde Decomposition of Toeplitz Matrices} \label{sec:VDint}

We present the FS Vandermonde decomposition result in this section. To encode the interval information into the Vandermonde decomposition, we first construct a trigonometric polynomial that is nonnegative on the interval $\cI$ and negative on its complement. We first clarify some notations. For $f_L\neq f_H\in\bT$, if $f_L< f_H$, then $\cI=\mbra{f_L, f_H}$ denotes a closed interval as usual. Otherwise, we define $\cI=\mbra{f_L, f_H}\coloneqq \bT\backslash \sbra{f_H, f_L}$. By this definition, we can conveniently deal with the case in which $0$ (or $1$) is an interior point of $\cI$. The trigonometric polynomial, $g$, is defined as:
\equ{g(z) = \frac{1}{z\sqrt{z_Lz_H}}\sbra{z-z_L}\sbra{z-z_H}\sgn\sbra{f_H-f_L}, \label{eq:g1}}
where $z_L \coloneqq e^{i2\pi f_L}$, $z_H \coloneqq e^{i2\pi f_H}$ and $\sgn\sbra{\cdot}$ is the sign function. With simple derivations, we have
\equ{g(z)= r_1 z^{-1} + r_0 + \overline{r}_1 z, \label{eq:gz1}}
where
{\lentwo\equa{r_0
&=& - \frac{z_L+z_H}{\sqrt{z_Lz_H}} \sgn\sbra{f_H-f_L} \notag \\
&=& -2\cos\mbra{\pi\sbra{f_H-f_L}}\sgn\sbra{f_H-f_L}, \label{eq:r0}\\ r_1
&=& \sqrt{z_Lz_H}\sgn\sbra{f_H-f_L} \notag\\
&=&  e^{i\pi\sbra{f_L+f_H}} \sgn\sbra{f_H-f_L}.\label{eq:r1}
}}It is evident that $g(z)$ is a Hermitian trigonometric polynomial that is real-valued on $\bT$. By the way that $g(z)$ is constructed, we know that $g(z)$ has two single roots $z_L$ and $z_H$, and equivalently, $g(f)$ has two single roots $f_L$ and $f_H$. Therefore, $g(f)$ flips its sign around $f_L$ and $f_H$. Two possibilities are: $g(f)$ is positive on $\sbra{f_L, f_H}$ and negative on $\sbra{f_H, f_L}$, or negative on $\sbra{f_L, f_H}$ and positive on $\sbra{f_H, f_L}$. To determine which one is true, we check the value at $f=\frac{1}{2}\sbra{f_L+f_H}$:
\equ{\begin{split}
&g\sbra{\frac{1}{2}\sbra{f_L+f_H}} \\
&= r_0 + 2\Re \sbra{r_1 e^{-i\pi \sbra{f_L+f_H} }} \\
&= \lbra{2-2\cos\mbra{\pi\sbra{f_L-f_H}}} \sgn\sbra{f_H-f_L}. \end{split}}
Consequently, the sign of $g$ at $f=\frac{1}{2}\sbra{f_L+f_H}$ is identical to that of $f_H-f_L$, meaning that $g(f)$ is always positive on $\sbra{f_L, f_H}$ and negative on $\sbra{f_H, f_L}$ whenever $f_L<f_H$ or $f_L>f_H$.

Now we are ready to present the FS Vandermonde decomposition result, which is summarized in the following theorem.\footnote{Part of the FS Vandermonde decomposition result was extended to a general form in the recent preprint \cite{chao2016semidefinite}, which appeared online after our conference paper \cite{yang2016vandermonde_ccc} was accepted.}

\begin{thm} Given $\cI\subset\bT$, a Toeplitz matrix $\m{T}\in\bC^{N\times N}$ admits an FS Vandermonde decomposition, as in \eqref{eq:VD}, with $f_k\in\cI$, if and only if
{\lentwo\equa{\m{T}
&\geq& \m{0}, \label{eq:Tupsd}\\ \m{T}_g
&\geq& \m{0}, \label{eq:Tunew}
}}where $g$ is defined by \eqref{eq:gz1}-\eqref{eq:r1} and $\m{T}_g$ by \eqref{eq:Tg}.
Moreover, the decomposition is unique if either $\m{T}$ or $\m{T}_g$ is rank-deficient. \label{thm:BCVD}
\end{thm}

\begin{proof} We first show the ``if'' part. Consider the case of $r\leq N-1$. It then follows from \eqref{eq:Tupsd} and Theorem \ref{thm:VD} that $\m{T}$ admits a unique Vandermonde decomposition as in \eqref{eq:VD}. So, it suffices to show $f_k\in \cI$, $k=1,\dots,r$ under the additional condition \eqref{eq:Tunew}. To do so, note by \eqref{eq:VD} that
\equ{t_{n-m} = T_{mn} = \sum_{k=1}^r p_k e^{i2\pi (m-n) f_k}. \label{eq:tlk}}
It immediately follows that
\equ{\begin{split}\mbra{T_g}_{mn}
&= \sum_{j=-1}^1 r_j t_{n-m+j} \\
&= \sum_{j=-1}^1 r_j \sum_{k=1}^r p_k e^{i2\pi (m-n-j) f_k}\\
&= \sum_{k=1}^r p_k e^{i2\pi (m-n) f_k} \sum_{j=-1}^1 r_j e^{-i2\pi j f_k}\\
&= \sum_{k=1}^r p_k g\sbra{f_k} e^{i2\pi (m-n) f_k}, \end{split} \label{eq:Tgkl}}
and hence
\equ{\begin{split}\m{T}_g
&= \sum_{k=1}^r p_k g\sbra{f_k} \m{a}\sbra{N-1,f_k}\m{a}^H\sbra{N-1,f_k}\\
&= \m{A}\sbra{N-1,\m{f}}\diag\sbra{p_1g\sbra{f_1},\dots, p_rg\sbra{f_r}} \m{A}^H\sbra{N-1,\m{f}}, \end{split} \label{eq:T1ueq}}
where $\m{A}\sbra{N-1,\m{f}} \coloneqq \mbra{\m{a}\sbra{N-1,f_1},\dots,\m{a}\sbra{N-1,f_r}}$ is an $\sbra{N-1}\times r$ Vandermonde matrix and $\diag\sbra{p_1g\sbra{f_1},\dots, p_rg\sbra{f_r}}$ denotes a diagonal matrix with $p_kg\sbra{f_k}$, $k=1,\dots,r$ on the diagonal. Note that $\m{A}\sbra{N-1,\m{f}}$ has full column rank since $r\leq N-1$. Using \eqref{eq:T1ueq} and \eqref{eq:Tunew}, we have that
\equ{\begin{split}
&\diag\sbra{p_1g\sbra{f_1},\dots, p_r g\sbra{f_r}} = \m{A}^{\dag}\sbra{N-1,\m{f}} \m{T}_g \m{A}^{\dag H}\sbra{N-1,\m{f}} \geq \m{0}, \end{split}}
where $\cdot^{\dag}$ denotes the matrix pseudo-inverse operator.
This means that $p_kg\sbra{f_k}\geq 0$, and since $p_k>0$, we have $g\sbra{f_k}\geq 0$, $k=1,\dots,r$. By the property of $g(f)$, finally, we have $f_k\in\cI$, $k=1,\dots,r$.

We next consider the case of $r=N$ in which $\m{T}$ is positive definite. Let $f_N = f_L$ and $p_N = \sbra{\m{a}^H\sbra{f_N} \m{T}^{-1} \m{a}\sbra{f_N}}^{-1}>0$. Similar to that in the proof of Theorem \ref{thm:VD}, we define a new sequence $\m{t}'=\mbra{t'_j}, \;\abs{j}\leq N-1$ as in \eqref{eq:t1j}, which therefore satisfies \eqref{eq:Tu1}, \eqref{eq:Tu1psd} and \eqref{eq:Tu1rank}. Moreover, we have
\equ{\begin{split}\mbra{T_g\sbra{\m{t}'}}_{mn}
&= \sum_{j=-1}^1 r_j t'_{n-m+j} \\
&= \mbra{T_g}_{mn} - p_N g(f_N) e^{i2\pi (m-n) f_N}, \end{split}}
and hence
\equ{\m{T}_g\sbra{\m{t}'} = \m{T}_g  - p_Ng\sbra{f_N} \m{a}\sbra{N-1,f_N}\m{a}^H\sbra{N-1,f_N}. \label{eq:Tgt1}}
By \eqref{eq:Tunew} and the fact that $g\sbra{f_N}=g\sbra{f_L}=0$, we have
\equ{\m{T}_g\sbra{\m{t}'} = \m{T}_g \geq \m{0}. \label{eq:Tu1newpsd}}
Now consider $\m{T}\sbra{\m{t}'}$ that satisfies \eqref{eq:Tu1psd}, \eqref{eq:Tu1rank} and \eqref{eq:Tu1newpsd}. Following from the ``if'' part of Theorem \ref{thm:BCVD} in the case of $r\leq N-1$ that we just proved, $\m{T}\sbra{\m{t}'}$ admits a unique decomposition as in \eqref{eq:VD}, with $f_k\in\cI$, $k=1,\dots,r=N-1$. Therefore, it follows from \eqref{eq:Tu1} that
\equ{\m{T} = \m{T}\sbra{\m{t}'} + p_N\m{a}\sbra{f_N}\m{a}^H\sbra{f_N}}
has a decomposition as in \eqref{eq:VD}, with $f_k\in\cI$, $k=1,\dots,r=N$. So we complete the ``if'' part.

The ``only if'' part can be shown by similar arguments. In particular, given $\m{T}$ as in \eqref{eq:VD}, it is evident that \eqref{eq:Tupsd} holds. Moreover, \eqref{eq:Tunew} also holds, since we still have \eqref{eq:T1ueq}, in which $g\sbra{f_k}\geq 0$, $k=1,\dots,r$ by the property of $g$.

We finally shown the uniqueness under the additional condition that $\m{T}$ or $\m{T}_g$ is rank-deficient. When $\m{T}$ is rank-deficient, this is a direct consequence of Theorem \ref{thm:VD}. In the other case when $\m{T}$ has full rank and $\m{T}_g$ is rank-deficient, note first that there are at least $N$ distinct $f_k$'s in the FS Vandermonde decomposition of $\m{T}$, since, otherwise, $\m{T}$ loses rank. We now recall \eqref{eq:T1ueq}, in which $\m{A}\sbra{N-1,\m{f}}$ has full row rank and $g(f_k)\geq 0$. To guarantee that $\m{T}_g$ is rank-deficient, $g(f_k)\neq 0$ must hold for maximally $N-2$, $f_k$'s and the other $f_k$'s must be either $f_L$ or $f_H$. This means that the decomposition consists of exactly $N$ atoms and two of them are located at $f_L$ and $f_H$. Therefore, the other $N-2$ frequencies are fixed as well, and the FS Vandermonde decomposition is unique.
\end{proof}

The FS Vandermonde decomposition can be computed similarly as the standard Vandermonde decomposition provided that the conditions of Theorem \ref{thm:BCVD} are satisfied. More concretely, in the case when $\m{T}$ is rank-deficient, it admits a unique Vandermonde decomposition that can be computed as in Section \ref{sec:standardVD}. In the case when $\m{T}$ has full rank, an $N$-atomic decomposition can be computed following from the proof of Theorem \ref{thm:BCVD}, to be specific, fix $f_N=f_L$ first and compute the other parameters following the proof.

Finally, note that the FS Vandermonde decomposition result can be extended straightforwardly to the multiple frequency band case. Let $K=\bigcup_{l=1}^J \mbra{f_{Ll}, f_{Hl}}$, where $\mbra{f_{Ll}, f_{Hl}}\subset\bT$, $l=1,\dots, J\geq 2$ are disjoint. We have the following corollary of Theorem \ref{thm:BCVD}, the proof of which is straightforward and thus is omitted.
\begin{cor} Given $K = \bigcup_{l=1}^J \mbra{f_{Ll}, f_{Hl}}$, a Toeplitz matrix $\m{T}\in\bC^{N\times N}$ admits an FS Vandermonde decomposition, as in \eqref{eq:VD}, with $f_k\in K$, if and only if there exist sequences $\m{t}_l$, $l=1,\dots,J$ satisfying
{\lentwo\equa{\sum_{l=1}^J \m{t}_l
&=& \m{t}, \label{eq:sumTleqT}\\  \m{T}\sbra{\m{t}_l}
&\geq& \m{0}, \label{eq:Tupsdl}\\ \m{T}_{g_l}\sbra{\m{t}_l}
&\geq& \m{0}, \quad l=1,\dots,J, \label{eq:Tunewl}
}}where $g_l$, $l=1,\dots,J$ are $g$ defined with respect to $\mbra{f_{Ll}, f_{Hl}}$, respectively. \label{cor:BCVDM}
\end{cor}


\section{Duality} \label{sec:duality}
Using the FS Vandermonde decomposition result presented in the previous section, we can explicitly characterize the cone of Toeplitz matrices admitting such decompositions. Due to the interest in optimization problems, we naturally look at the dual cone, which, as we will see, enables us to link the FS Vandermonde decomposition to the theory of trigonometric polynomials, to be specific, the PRL given in \cite{davidson2002linear,alkire2002convex} (see also \cite{dumitrescu2007positive}).

For a sequence $\m{t}=\mbra{t_j}$, $\abs{j}\leq N-1$ with $t_{-j} = \overline{t}_j$, let $\m{t}_R = \mbra{\Re t_{N-1}, \dots, \Re t_{1}, \frac{\sqrt{2}}{2}t_0, \Im t_1, \dots, \Im t_{N-1}}^T\in\bR^{2N-1}$ be a representation of $\m{t}$ in the real domain, where the coefficient $\frac{\sqrt{2}}{2}$ for $t_0$ is chosen for convenience. It is obvious that all $N\times N$ Toeplitz matrices admitting an FS Vandermonde decomposition on a given interval $\cI\subset\bT$ form a cone that can be identified with
\equ{\begin{split}
&\cK_{\text{VDF}} \coloneqq \lbra{\m{t}_R:\; \m{T} = \sum_{k} p_k\m{a}\sbra{f_k}\m{a}^H\sbra{f_k},\; p_k\geq 0, f_k\in \cI}.\end{split} \label{eq:KVDF}}
Define
\equ{\begin{split}\cK_{\text{VDM}}
&\coloneqq \lbra{\m{t}_R:\; \m{T}\geq\m{0}, \; \m{T}_{g}\geq \m{0}}, \end{split} \label{eq:KVDM}}
where $g$ is defined in Theorem \ref{thm:BCVD}. A direct consequence of Theorem \ref{thm:BCVD} is that
\equ{\cK_{\text{VDF}} = \cK_{\text{VDM}}. \label{eq:VDconeeq} }

We next consider the dual cone of $\cK_{\text{VDF}}$ defined as \cite{boyd2004convex}
\equ{\cK_{\text{VDF}}^*\coloneqq \lbra{\m{\alpha}\in\bR^{2N-1}:\; \m{t}_R^T\m{\alpha}\geq 0 \text{ for any } \m{t}_R\in\cK_{\text{VDF}}}. \label{eq:defdualcone}}
Before proceeding to the main result of this section, we first introduce some notations. Let
\equ{\cK_{\text{PolF}}\coloneqq \lbra{\m{\gamma}_R:\; \sum_{j=1-N} ^{N-1} \gamma_j e^{i2\pi j f} \geq 0,\; f\in \cI} \label{KPolF}}
denote the cone of trigonometric polynomials of order $N-1$ and nonnegative on $\cI$, where $\m{\gamma}_R$ is similarly defined as $\m{t}_R$. Let also $\m{\Theta}_j$, $\abs{j}\leq N-1$ be an $N\times N$ elementary Toeplitz matrix with ones on its $j$th diagonal and zeros elsewhere. With respect to $\m{\Theta}_j$ and the trigonometric polynomial $g$ defined by \eqref{eq:gz1}-\eqref{eq:r1}, we define the $(N-1)\times (N-1)$ Toeplitz matrix $\m{\Theta}_{gj}$, like $\m{T}_g$ with respect to $\m{T}$. By definition, it is easy to verify that
\lentwo{\equa{\m{T}
&=& \sum_{j=1-N}^{N-1} \m{\Theta}_j t_j, \label{eq:TinTheta}\\ \m{T}_g
&=& \sum_{j=1-N}^{N-1} \m{\Theta}_{gj} t_j. \label{eq:TginThetag}
}}We also define the cone
\equ{\begin{split} \cK_{\text{PolM}}\coloneqq \Big\{\m{\gamma}_R:\;
& \gamma_{-j} = \tr\sbra{\m{\Theta}_j \m{Q}_0} + \tr\mbra{ \m{\Theta}_{gj} \m{Q}_1 }, \\
& \abs{j}\leq N-1, \\
& \m{Q}_0\in\bC^{N\times N}, \m{Q}_1\in\bC^{(N-1)\times (N-1)},\\
& \m{Q}_0\geq\m{0},\m{Q}_1\geq\m{0} \Big\}. \end{split} \label{KPolM}}
The main result of this section is given in the following theorem.

\begin{thm} We have the following identities:
 {\lentwo\equa{\cK_{\text{VDF}}^*
 &=& \cK_{\text{PolF}}, \label{eq:dualeq1} \\ \cK_{\text{PolM}}^*
 &=& \cK_{\text{VDM}}. \label{eq:dualeq2}
 }}Therefore, provided that $\cK_{\text{VDF}} = \cK_{\text{VDM}}$ we can conclude that $\cK_{\text{PolF}} = \cK_{\text{PolM}}$, and vice versa.
\label{thm:dualcone}
\end{thm}

\begin{proof} We first show \eqref{eq:dualeq1}. Note that $\m{t}_R\in\cK_{\text{VDF}}$ if and only if
\equ{t_j = \sum_{k} p_k e^{-i2\pi j f_k}, \quad j = 1-N,\dots,N-1, \label{eq:tj_KVDF}}
where $p_k\geq 0$ and $f_k\in\cI$. For any $\m{\alpha}=\mbra{\alpha_{1-N},\dots,\alpha_{N-1}}^T\in\bR^{2N-1}$, we define $\m{\gamma}\in\bC^{2N-1}$ such that $\gamma_0 = \sqrt{2}\alpha_0$, $\gamma_j = \alpha_{-j}+i\alpha_{j}$ and $\gamma_{-j} = \alpha_{-j}-i\alpha_{j}$, $j=1,\dots,N-1$. It follows that $\m{\alpha} = \m{\gamma}_R$ and
\equ{\begin{split}\m{t}_R^T\m{\alpha}
&= \frac{\sqrt{2}}{2}t_0\cdot \frac{\sqrt{2}}{2}\gamma_0 + \Re \sum_{j=1}^{N-1}  \overline{t}_j \gamma _j\\
&= \frac{1}{2} \sum_{j=1-N}^{N-1} \overline{t}_j\gamma_j. \end{split} \label{eq:realtHgamma1}}
Inserting \eqref{eq:tj_KVDF} into \eqref{eq:realtHgamma1}, we have that
\equ{\m{t}_R^T\m{\alpha} = \frac{1}{2} \sum_{k} p_k \sum_{j=1-N}^{N-1} \gamma_j e^{i2\pi j f_k}. \label{eq:realtHgamma}}
By \eqref{eq:realtHgamma} and the definition of the dual cone, $\m{\alpha} = \m{\gamma}_R \in\cK_{\text{VDF}}^*$ if and only if the right hand of \eqref{eq:realtHgamma} is nonnegative for any $p_k\geq 0$ and any $f_k\in\cI$. The above condition holds if and only if $h(f) \coloneqq \sum_{j=1-N}^{N-1} \gamma_j e^{i2\pi j f}$ is nonnegative on $\cI$, or equivalently, $\m{\alpha} \in \cK_{\text{PolF}}$ by \eqref{KPolF}.

To show \eqref{eq:dualeq2}, we can similarly define $\m{t}$ for $\m{\alpha}\in\bR^{2N-1}$ such that $\m{\alpha} = \m{t}_R$. It follows that $\m{T}$ and $\m{T}_g$ are Hermitian. For any $\m{\gamma}_R\in \cK_{\text{PolM}}$, which can be expressed as in \eqref{KPolM}, we have that
\equ{\begin{split}\m{\gamma}_R^T\m{\alpha}
&= \frac{1}{2} \sum_{j=1-N}^{N-1} \overline{\gamma}_j t_j\\
&= \frac{1}{2} \sum_{j=1-N}^{N-1} \gamma_{-j} t_j \\
&= \frac{1}{2} \sum_{j=1-N}^{N-1} t_j \lbra{\tr\sbra{\m{\Theta}_j \m{Q}_0} + \tr\mbra{ \m{\Theta}_{gj} \m{Q}_1 }}. \end{split}}
Using the identities in \eqref{eq:TinTheta} and \eqref{eq:TginThetag}, we have that
\equ{\m{\gamma}_R^T\m{\alpha} = \frac{1}{2} \tr\sbra{\m{T}\m{Q}_0} + \frac{1}{2} \tr\sbra{\m{T}_g \m{Q}_1}. \label{eq:equ_dualcone}}
By the definition of the dual cone, $\m{\alpha} \in \cK_{\text{PolM}}^*$ if and only if $\m{\gamma}_R^T\m{\alpha}\geq 0$ for any $\m{\gamma}_R\in \cK_{\text{PolM}}$. Using \eqref{KPolM} and \eqref{eq:equ_dualcone}, the above condition holds if and only if $\tr\sbra{\m{T}\m{Q}_0} + \tr\sbra{\m{T}_g \m{Q}_1}\geq 0$ for any $\m{Q}_0\geq \m{0}$ and $\m{Q}_1\geq \m{0}$, which holds if and only if $\tr\sbra{\m{T}\m{Q}_0}\geq 0$ for any $\m{Q}_0\geq \m{0}$ and $\tr\sbra{\m{T}_g \m{Q}_1}\geq 0$ for any $\m{Q}_1\geq \m{0}$, and is further equivalent to the condition $\m{T}\geq \m{0}$ and $\m{T}_g\geq \m{0}$. The last condition is equivalent to  $\m{\alpha} = \m{t}_R\in \cK_{\text{VDM}}$ by \eqref{eq:KVDM}.

Finally, provided that $\cK_{\text{VDF}} = \cK_{\text{VDM}}$ and using \eqref{eq:dualeq1} and \eqref{eq:dualeq2}, we have that
\equ{\cK_{\text{PolF}} = \cK_{\text{VDF}}^* = \cK_{\text{VDM}}^* = \cK_{\text{PolM}}^{**}. }
Using the identify that $\cK_{\text{PolM}}^{**} = \cK_{\text{PolM}}$, which follows from the fact that $\cK_{\text{PolM}}$ is convex and closed \cite{boyd2004convex}, we conclude that $\cK_{\text{PolF}} = \cK_{\text{PolM}}$.
By similar arguments we can also show that $\cK_{\text{VDF}} = \cK_{\text{VDM}}$ provided that $\cK_{\text{PolF}} = \cK_{\text{PolM}}$.
\end{proof}

By Theorem \ref{thm:dualcone}, the FS Vandermonde decomposition on $\cI$ is linked via duality to the trigonometric polynomials nonnegative on the same interval. Moreover, the identity that $\cK_{\text{PolF}} = \cK_{\text{PolM}}$ provides a matrix form parametrization of the coefficients of these polynomials. In fact, this is exactly the Gram matrix parametrization concluded by the PRL in \cite{davidson2002linear,alkire2002convex} (see also \cite{dumitrescu2007positive}). This means that the PRL in \cite{davidson2002linear,alkire2002convex} can be obtained from the FS Vandermonde decomposition; conversely, the PRL also provides an alternative way to characterize the set of Toeplitz matrices admitting an FS Vandermonde decomposition.\footnote{Note that Theorem \ref{thm:BCVD} is stronger in the sense that it concludes that all such Toeplitz matrices always admit a decomposition containing $N$ atoms or less.} Therefore, it will not be surprising that, as we will see, for certain convex optimization problems the two techniques can be applied to give the primal and the dual problems, respectively. But note that there are indeed scenarios in which one technique can be applied while the other cannot. Examples will be provided in the ensuing sections to demonstrate the usefulness of the FS Vandermonde decomposition.

\begin{rem} The trigonometric polynomial $g(z) = r_{-1}z + r_0 + r_1z^{-1}$ that is nonnegative on $\cI$ and negative on its complement plays an important role in both the FS Vandermonde decomposition of Toeplitz matrices and the Gram matrix parametrization of trigonometric polynomials. It is worth noting that the polynomial defined in the present paper (recall \eqref{eq:gz1}-\eqref{eq:r1}) is different from those in \cite{dumitrescu2007positive,davidson2002linear,alkire2002convex}. As a matter of fact, while the polynomial we define applies uniformly to all intervals $\cI\in\bT$, certain modifications to the polynomial or additional operations such as sliding the interval have to be taken in \cite{dumitrescu2007positive, davidson2002linear,alkire2002convex} when $\cI$ contains certain critical points such as $0$ (or $1$) and $\frac{1}{2}$.
\end{rem}

\section{Application in the Theory of Moments} \label{sec:moment}

\subsection{Problem Statement}
For a given sequence $t_j$, $\abs{j}\leq N-1$ and a given domain $F$, a truncated moment problem entails determining whether there exists a positive Borel measure $\mu$ on $F$ such that \cite{akhiezer1965classical}
\equ{t_j = \int_F z^j \text{d} \mu\sbra{z}, \quad \abs{j}\leq N-1. \label{eq:tj_moment}}
The problem is further referred to as a truncated $K$-moment problem if $\mu$ is constrained to be supported on a semialgebraic set $K\subset F$, i.e., \cite{schmudgen1991thek}
\equ{\supp\sbra{\mu} \subset K. \label{eq:supp_moment}}
A measure $\mu$ satisfying \eqref{eq:tj_moment} is a representing measure for $\m{t}$; $\mu$ is a $K$-representing measure if it satisfies \eqref{eq:tj_moment} and \eqref{eq:supp_moment}.

The truncated moment and $K$-moment problems have been solved when $F$ is the real or the complex domain (note that the complex moment problem is defined slightly differently from \eqref{eq:tj_moment}) \cite{curto1991recursiveness,curto2000truncated,lasserre2009moments}. The truncated moment problem is also solved when $F$ is the unit circle, known as the truncated trigonometric moment problem \cite{grenander1958toeplitz,curto1991recursiveness}. In fact, the solution is given by evoking the Vandermonde decomposition of Toeplitz matrices: A representing measure $\mu$ exists if and only if the Toeplitz matrix $\m{T}$ formed using $\m{t}$ admits a Vandermonde decomposition, or equivalently, $\m{T}\geq \m{0}$ by Theorem \ref{thm:VD}. To the best of our knowledge, however, the truncated trigonometric $K$-moment problem is still open. This section is devoted to a solution to this problem by applying the FS Vandermonde decomposition.

Note that a semialgebraic set $K$ on the unit circle $\bT$ can be identified with the union of finite disjoint subintervals $\mbra{f_{Ll}, f_{Hl}} \subset \bT$, $l=1,\dots,J$. Therefore, the moment problem of interest can be restated as follows. For a given sequence $t_j$, $\abs{j}\leq N-1$, the truncated trigonometric $K$-moment problem entails determining whether there exists a $K$-representing measure $\mu$ on $\bT$ satisfying that
{\lentwo\equa{t_j
&=& \int_{\bT} e^{-i2\pi jf} \text{d} \mu\sbra{f}, \quad \abs{j}\leq N-1, \label{eq:tj_moment1} \\ \supp\sbra{\mu}
&\subset& K=\bigcup_{l=1}^J \mbra{f_{Ll}, f_{Hl}}\subset \bT. \label{eq:supp_moment1}} }

\subsection{Proposed Solution}
Let $\m{T}$ be the $N\times N$ Toeplitz matrix formed using the moment sequence $t_j$, $\abs{j}\leq N-1$. Suppose that an $r$-atomic $K$-representing measure $\mu$ for $\m{t}$ exists that satisfies \eqref{eq:tj_moment1} and \eqref{eq:supp_moment1}. It follows from \eqref{eq:supp_moment1} that
\equ{\mu\sbra{f} = \sum_{k=1}^r p_k \delta_{f_k}, \quad f_k \in K, \label{eq:mut}}
where $\delta_{f}$ is the Dirac delta function and $p_k>0$ denotes the density at $f_k$. Inserting \eqref{eq:mut} into \eqref{eq:tj_moment1}, we have that
\equ{t_j= \sum_{k=1}^r p_k e^{-i2\pi jf_k}, \quad \abs{j}\leq N-1,\; f_k\in K.}
It follows that
\equ{\m{T} = \sum_{k=1}^r p_k \m{a}\sbra{f_k}\m{a}^H\sbra{f_k}, \quad f_k\in K.}
This means that $\m{T}$ admits an $r$-atomic FS Vandermonde decomposition on $K$. It is easy to show that the above arguments also hold conversely. So we conclude the following result.
\begin{lem} An $r$-atomic $K$-representing measure $\mu$ for $\m{t}$ exists if and only if $\m{T}$ admits an $r$-atomic FS Vandermonde decomposition on $K$. \label{lem:moment_VD}
\end{lem}

We next provide explicit conditions on $\m{T}$ by applying Theorem \ref{thm:BCVD}. In the case when $K$ is a single interval, the following theorem is a direct consequence by combining Lemma \ref{lem:moment_VD} and Theorem \ref{thm:BCVD}.

\begin{thm} Given $K=\mbra{f_L, f_H}$, an $r$-atomic $K$-representing measure $\mu$ for $\m{t}$ exists if and only if \eqref{eq:Tupsd} and \eqref{eq:Tunew} hold, where $r=\rank\sbra{\m{T}}$, and $g$ is defined by \eqref{eq:gz1}-\eqref{eq:r1}. Moreover, $\mu$ can be found by computing the FS Vandermonde decomposition of $\m{T}$ on $K$, and it is unique if $\m{T}$ or $\m{T}_g$ is rank-deficient. \label{thm:momentsolution1}
\end{thm}

In the multiple frequency band case in which $K=\bigcup_{l=1}^J \mbra{f_{Ll}, f_{Hl}}$, corresponding to Corollary \ref{cor:BCVDM}, we have the following corollary of Theorem \ref{thm:momentsolution1}. The proof is trivial and is omitted.
\begin{cor} Given $K = \bigcup_{l=1}^J \mbra{f_{Ll}, f_{Hl}}$, a $K$-representing measure $\mu$ for $\m{t}$ exists if and only if
there exist sequences $\m{t}_l$, $l=1,\dots,J$ satisfying \eqref{eq:sumTleqT}-\eqref{eq:Tunewl}. \label{cor:momentsolutionM}
\end{cor}

Corollary \ref{cor:momentsolutionM} provides a numerical approach to finding a $K$-representing measure, if it exists, by solving the following feasibility problem that is a SDP:
\equ{\begin{split}
&\text{Find}\; \m{t}_l, \quad l=1,\dots,J,\\
&\text{subject to \eqref{eq:sumTleqT}-\eqref{eq:Tunewl}}. \end{split} \label{eq:feasprob}}
If a solution, denoted by $\m{t}_l^*$, $l=1,\dots,J$, can be found, then we can find representing measures for $\m{t}_l^*$ on each corresponding interval by Theorem \ref{thm:momentsolution1}, the sum of which finally form a $K$-representing measure for $\m{t}$. If \eqref{eq:feasprob} is infeasible, then no $K$-representing measure for $\m{t}$ exists.

\begin{rem} In the case when $\m{T}$ has full rank, the representing measure $\mu$ might not be unique, if it exists. By solving \eqref{eq:feasprob}, we actually find one among them. In this case the obtained measure $\mu$ may consist of as large as $NJ$ atoms. To possibly reduce the number of atoms (a.k.a. to simplify the obtained measure), we can find the one minimizing certain convex function of $\m{t}_l$, $l=1,\dots,J$, e.g., $\pm\tr\sbra{\m{T}\sbra{\m{t}_1}}$. By doing so, it is expected that certain $\m{T}\sbra{\m{t}_l}$'s are rank-deficient and thus result in a small number of atoms. \label{rem:simplifymeasure}
\end{rem}

Finally, it is interesting to note that the dual problem of \eqref{eq:feasprob} can be easily obtained using the result in Section \ref{sec:duality}. Using the cone notations \eqref{eq:feasprob} can be written as:
\equ{\begin{split}
&\text{Find}\; \m{t}_{R,l}\in \cK_{\text{VDM},l}, \quad l=1,\dots,J,\\
&\text{subject to } \sum_{l=1}^J \m{t}_{R,l} = \m{t}_R, \end{split} \label{eq:feasprob2}}
where $\m{t}_{R,l}\coloneqq \mbra{\m{t}_l}_R$, and $\cK_{\text{VDM},l}$ denotes $\cK_{\text{VDM}}$ in \eqref{eq:KVDM} with $g$ being $g_l$. The Lagrangian function is given by:
\equ{\begin{split}\cL\sbra{\m{t}_{R,1},\dots,\m{t}_{R,J}, \m{\alpha}}
&= \sbra{\sum_{l=1}^J \m{t}_{R,l} - \m{t}_R}^T\m{\alpha} \\
&= \sum_{l=1}^J \m{t}_{R,l}^T\m{\alpha} - \m{t}_R^T\m{\alpha}, \end{split}}
where $\m{t}_{R,l}\in \cK_{\text{VDM},l}$, $l=1,\dots,J$, and $\m{\alpha}$ is the Lagrangian multiplier. Using the knowledge of the dual cone, we have that
\equ{\min_{\m{t}_{R,l} \in \cK_{\text{VDM},l}} \cL = \left\{ \begin{array}{ll} -\m{t}_R^T\m{\alpha}, & \text{if } \m{\alpha}\in \cK_{\text{VDM},l}^*, l=1,\dots,J;\\ -\infty, & \text{otherwise.} \end{array} \right. }
Therefore, the dual problem is given by:
\equ{\max_{\m{\alpha}} \m{t}_R^T\m{\alpha}, \st \m{\alpha}\in \bigcap_{l=1}^J \cK_{\text{PolM},l}, \label{eq:dualprob}}
where we have used the identity that $\cK_{\text{VDM},l}^*=\cK_{\text{PolM},l}$ given by Theorem \ref{thm:dualcone}. Note that \eqref{eq:dualprob} can be cast as SDP following from \eqref{KPolM}.

\begin{exa} Suppose that the moment sequence $t_j$, $\abs{j}\leq N-1$ is generated from its $3$-atomic representing measure
\equ{\mu_1 = 0.7\delta_{0.1} + 2\delta_{0.25} + \delta_{0.7},}
which is plotted in Fig. \ref{Fig:moment} together with $\mu_j$, $j=2,\dots,5$ that will be solved for.
\begin{itemize}
 \item[1)] In the case of $N\geq 4$, we can form the Toeplitz matrix $\m{T}$ using $\m{t}$, having that $\rank\sbra{\m{T}} = 3 < N$. By Theorem \ref{thm:VD}, $\mu_1$ is the unique representing measure for $\m{t}$.
 \item[2)]Suppose that $N=3$ and $K = \mbra{0.05, 0.75}$. Since $K$ includes all the frequencies in $\mu_1$, one representing measure on $K$ has already been given by $\mu_1$. By the existence of the representing measure, it follows from Theorem \ref{thm:momentsolution1} that $\m{T}$ and $\m{T}_g$ are both PSD. Applying the proposed FS Vandermonde decomposition algorithm to the solution, the following $3$-atomic $K$-representing measure is obtained:
     \equ{\mu_2 =0.4630\delta_{0.05} + 2.2485\delta_{0.2383} + 0.9885\delta_{0.6927}, \notag}
     which is somehow similar to $\mu_1$. Note that the frequency $0.05$ in $\mu_2$ is nothing but the staring point of $K$, which has been deliberately chosen in the presented decomposition algorithm. Note that
 \item[3)]Suppose that $N=3$ and $K = \mbra{0.05, 0.3}\cup \mbra{0.65,0.75}$. One representing measure for $\m{t}$ is also given by $\mu_1$. To possibly find another one, we solve \eqref{eq:feasprob} using SDPT3 \cite{toh1999sdpt3} in Matlab and a solution is successfully found. Applying FS Vandermonde decomposition to the solution, a $6$-atomic $K$-representing measure is given by:
     \equ{\begin{split}\mu_3
     &= 0.1825\delta_{0.05} + 1.2284\delta_{0.1764}\\
     &\quad+ 1.2713\delta_{0.2722}+ 0.1546\delta_{0.65} \\
     &\quad + 0.5088\delta_{0.6917} + 0.3545\delta_{0.7436}. \end{split} \notag}
     In $\mu_3$, $0.05$ and $0.65$ are the starting points of the two intervals of $K$. The first three frequencies are located on the first interval and the other three frequencies are on the other interval.
  \item[4)] Suppose that $N=3$. We want to check whether one representing measure exists on $K = \mbra{0.2, 0.3} \cup \mbra{0.6, 0.8}$. To do so, we also solve \eqref{eq:feasprob} and a solution is successfully found. This means that a $K$-representing measure exists for $\m{t}$ by Corollary \ref{cor:momentsolutionM}. Applying the FS Vandermonde decomposition, a $6$-atomic $K$-representing measure is given by:
     \equ{\begin{split}\mu_4
     &= 1.9614\delta_{0.2} + 0.1296\delta_{0.2290}\\
     &\quad+ 0.4456\delta_{0.2891}+ 0.2437\delta_{0.6} \\
     &\quad + 0.3637\delta_{0.6467} + 0.5561\delta_{0.7962}. \end{split} }
 \item[5)] With the same settings as in 4), instead of solving \eqref{eq:feasprob}, we find the one maximizing $\tr\sbra{\m{T}\sbra{\m{t}_1}}$ among all feasible representing measures on $K$, following Remark \ref{rem:simplifymeasure}. The obtained solution $\sbra{\m{t}_1^*, \m{t}_2^*}$ satisfies that $\rank\sbra{\m{T}\sbra{\m{t}_1^*}}=\rank\sbra{\m{T}\sbra{\m{t}_2^*}}=2<N$, resulting in the following $4$-atomic representing measure:
     \equ{\begin{split} \mu_5
      &= 2.0837\delta_{0.2} + 0.4726\delta_{0.3}\\
      &\quad + 0.6218\delta_{0.6382} + 0.5219\delta_{0.8}. \end{split}}
     Compared to $\mu_4$, the number of atoms of $\mu_5$ is reduced.
 \item[6)] Suppose that $N=3$ and $K = \mbra{0.2, 0.3} \cup \mbra{0.6, 0.75}$. Then \eqref{eq:feasprob} is infeasible. This means that no $K$-representing measure for $\m{t}$ exists by Corollary \ref{cor:momentsolutionM}.
\end{itemize}
\end{exa}

\begin{figure*}
\centering
      \includegraphics[width=4in]{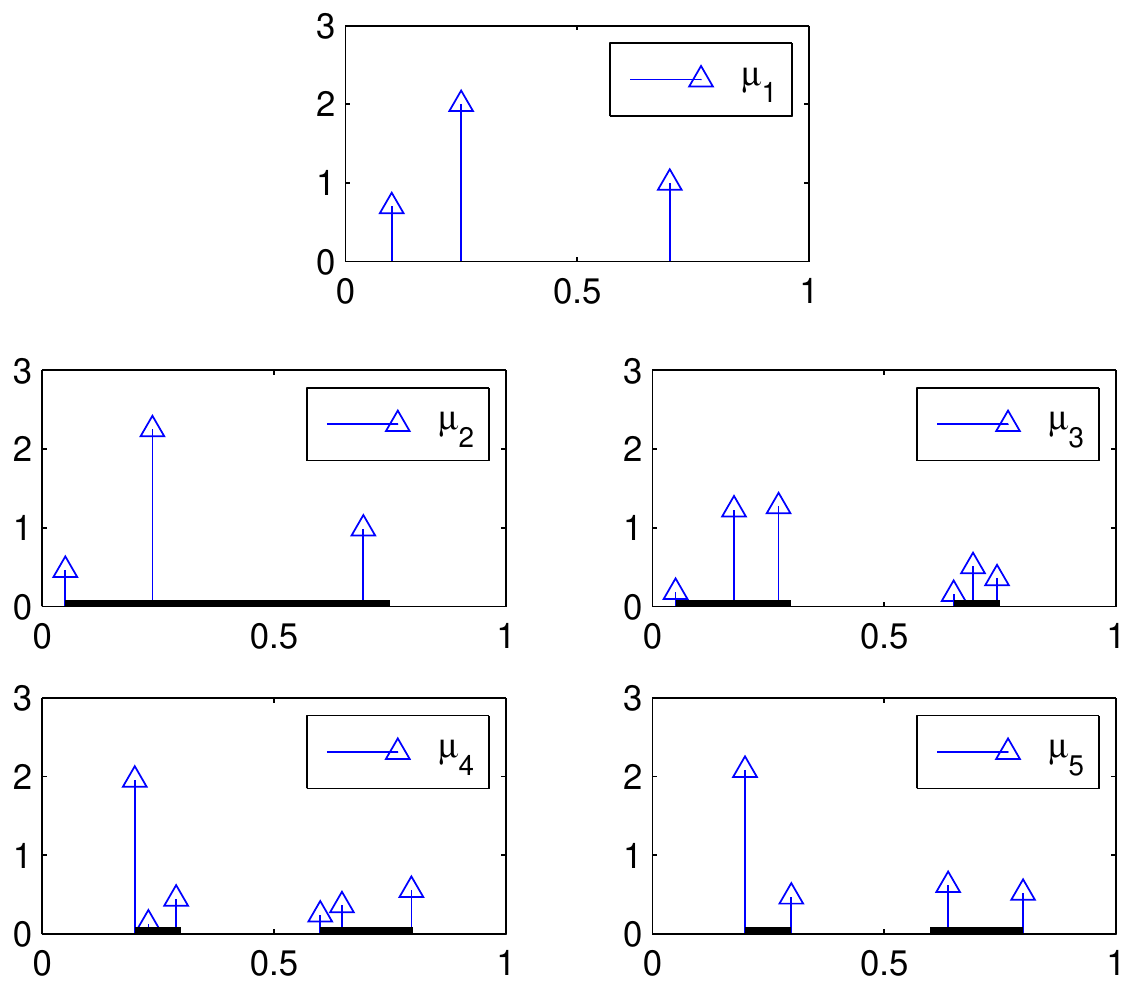}
\centering
\caption{Solved representing measures $\mu_j$, $j=2,\dots,5$ given a moment sequence generated from $\mu_1$ and a semialgebraic set $K$ (indicated by the line segments on the $x$-axis).} \label{Fig:moment}
\end{figure*}

\section{Application in Line Spectral Estimation} \label{sec:application}

\subsection{Problem Statement}
Line spectral estimation can be found in wide applications such as communications, radar, sonar, and so on \cite{stoica2005spectral}. In particular, we have the following data model in the absence of noise:\footnote{Note that the noisy case can be dealt with similarly with minor modifications on the presented solution. Discussions will be provided later.}
\equ{\m{y}^o = \sum_{k=1}^r \m{a}\sbra{f_k}s_k = \m{A}\sbra{\m{f}}\m{s}, \label{eq:model1}}
where $\m{y}^o\in\bC^N$ is a uniformly sampled signal (at a Nyquist rate), $f_k\in\bT$ and $s_k\in\bC$ are the normalized frequency and the complex amplitude of the $k$th sinusoid respectively, and $r$ is the number of sinusoids. To estimate the frequencies, we are given a part of the entries of $\m{y}^o$ that form the subvector $\m{y}_{\Omega}^o\in\bC^M$, where $\Omega$ denotes the set of sampling indexes and is of cardinality $M<N$. This frequency estimation problem is referred to as off-grid/continuous compressed sensing in \cite{tang2012compressed} in the sense that we have compressive data as in the pioneering work of compressed sensing \cite{candes2006robust}, but differently, the frequencies can take any continuous value in $\bT$ as opposed to the discrete setting in \cite{candes2006robust}.

In this section, we consider the case when the frequencies are known {\em a priori} to lie in an interval $\cI\subset\bT$. Inspired by the recent atomic norm techniques \cite{candes2013towards,candes2013super, bhaskar2013atomic, tang2012compressed,yang2015gridless}, the paper \cite{mishra2015spectral} proposed an FS atomic norm approach (or constrained atomic norm in the language of \cite{mishra2015spectral}) that was shown to achieve better performance than the standard atomic norm by exploiting the prior knowledge. In particular, define the (FS) set of atoms
\equ{\cA\sbra{\cI} \coloneqq \lbra{\m{a}\sbra{f_k,\phi_k}=\m{a}\sbra{f}\phi:\; f\in\cI, \;\abs{\phi}=1}. \label{eq:A_I}}
The FS atomic norm is the atomic norm induced by $\cA\sbra{\cI}$:
\equ{\begin{split}\norm{\m{y}}_{\cA\sbra{\cI}} \coloneqq
&\inf_{c_k>0,\m{a}_k\in\cA\sbra{\cI}}\lbra{\sum_k c_k:\; \m{y} = \sum_k c_k \m{a}_k}\\ =
&\inf_{f_k\in\cI,s_k}\lbra{\sum_k \abs{s_k}:\; \m{y} = \sum_k \m{a}\sbra{f_k}s_k}. \end{split} \label{eq:atomn_I}}
The following FS atomic norm minimization (FS-ANM) problem was proposed in \cite{mishra2015spectral}:
\equ{\min_{\m{y}} \norm{\m{y}}_{\cA\sbra{\cI}}, \st \m{y}_{\Omega} = \m{y}_{\Omega}^o. \label{eq:atomnmin_I}}
This means that, among all candidates $\m{y}$ which are consistent with the acquired samples $\m{y}_{\Omega}^o$, we find the one $\m{y}^*$ with the minimum FS atomic norm as the signal estimate, and the frequencies composing $\m{y}^*$ form the frequency estimates. Note that the noisy case can be dealt with similarly following a standard routine (by replacing the equality constraint in \eqref{eq:atomnmin_I} by $\twon{\m{y}_{\Omega} - \m{y}_{\Omega}^o}\leq \eta$ given the upper bound $\eta$ on the noise energy).
Note also that \eqref{eq:A_I}-\eqref{eq:atomnmin_I} degenerate to the existing standard forms in the case of $\cI=\bT$.

Since the FS atomic norm defined in \eqref{eq:atomn_I} is inherently semi-infinite programming (SIP), a finite-dimensional formulation of it is required to practically solve \eqref{eq:atomnmin_I}, which is dealt with in the ensuing section by applying the FS Vandermonde decomposition.

\subsection{SDP Formulation of FS Atomic Norm}

By applying the FS Vandermonde decomposition, the FS atomic norm is cast as SDP in the following theorem.
\begin{thm} It holds that
\equ{\begin{split} \norm{\m{y}}_{\cA\sbra{\cI}}
=& \min_{x, \m{t}} \frac{1}{2}x + \frac{1}{2}t_0,\\
& \st \begin{bmatrix} x & \m{y}^H \\ \m{y} & \m{T} \end{bmatrix}\geq \m{0} \text{ and } \m{T}_g\geq \m{0}, \end{split} \label{eq:atomnDsdp}}
where $g$ is as defined previously.
 \label{thm:atomnD}
\end{thm}

\begin{proof} Let $F^*$ be the optimal objective value of \eqref{eq:atomnDsdp}. We need to show that $\norm{\m{y}}_{\cA\sbra{\cI}} = F^*$.

We first show that $F^* \leq \norm{\m{y}}_{\cA\sbra{\cI}}$. To do so, let $\m{y} = \sum_k c_k\m{a}\sbra{f_k,\phi_k}$ be an FS atomic decomposition of $\m{y}$ on $\cI$. Then let $\m{t}$ be such that $\m{T}(\m{t}) = \sum_k c_k \m{a}\sbra{f_k}\m{a}^H\sbra{f_k}$ and $x = \sum_k c_k$. By Theorem \ref{thm:BCVD}, we have that $\m{T}_g\geq\m{0}$. Moreover, it holds that
\equ{\begin{bmatrix} x & \m{y}^H \\ \m{y} & \m{T} \end{bmatrix} = \sum_k c_k\begin{bmatrix} \overline{\phi}_k \\ \m{a}\sbra{f_k} \end{bmatrix} \begin{bmatrix} \overline{\phi}_k \\ \m{a}\sbra{f_k} \end{bmatrix}^H \geq \m{0}.}
Therefore, $x$ and $\m{t}$ constructed above form a feasible solution to the problem in \eqref{eq:atomnDsdp}, at which the objective value equals
\equ{\frac{1}{2}x + \frac{1}{2}t_0 = \sum_k c_k.}
It follows that $F^* \leq \sum_k c_k$. Since the inequality holds for any FS atomic decomposition of $\m{y}$ on $\cI$, we have that $F^* \leq \norm{\m{y}}_{\cA\sbra{\cI}}$ by the definition of $\norm{\m{y}}_{\cA\sbra{\cI}}$.

On the other hand, suppose that $\sbra{x^*, \m{t}^*}$ is an optimal solution to the problem in \eqref{eq:atomnDsdp}. By the fact that $\m{T}(\m{t}^*)\geq \m{0}$ and $\m{T}_g(\m{t}^*)\geq \m{0}$ and applying Theorem \ref{thm:BCVD}, we have that $\m{T}\sbra{\m{t}^*}$ has an FS Vandermonde decomposition on $\cI$ as in \eqref{eq:VD} with $\sbra{r, p_k, f_k}$ denoted by $\sbra{r^*, p_k^*, f_k^*}$.
By the fact that $\begin{bmatrix} x^* & \m{y}^H \\ \m{y} & \m{T}\sbra{\m{t}^*} \end{bmatrix}\geq\m{0}$, we have that $\m{y}$ lies in the range space of $\m{T}\sbra{\m{t}^*}$ and thus has the following FS atomic decomposition:
\equ{\m{y} = \sum_{k=1}^{r^*} c_k^*\m{a}\sbra{f_k^*,\phi_k^*}, \quad f_k^*\in\cI. \label{eq:yatomdec}}
Moreover, it holds that
{\lentwo\equa{x^*
&\geq& \m{y}^H \mbra{\m{T}\sbra{\m{t}^*}}^{\dag}\m{y} = \sum_{k=1}^{r^*} \frac{c_k^{*2}}{p_k^*},\\ t_0^*
&=& \sum_{k=1}^{r^*} p_k^*.
}}It therefore follows that
\equ{\begin{split}F^*
&= \frac{1}{2}x^* + \frac{1}{2}t_0^* \\
&\geq \frac{1}{2}\sum_k \frac{c_k^{*2}}{p_k^*} + \frac{1}{2}\sum_k p_k^* \\
&\geq \sum_{k} c_k^*\\
&\geq \norm{\m{y}}_{\cA\sbra{\cI}}. \end{split} \label{eq:Fleqsum} }
Combining \eqref{eq:Fleqsum} and the inequality that $F^* \leq \norm{\m{y}}_{\cA\sbra{\cI}}$ as shown previously, we conclude that $F^* = \norm{\m{y}}_{\cA\sbra{\cI}}$ and complete the proof. At last, it is worth noting that by \eqref{eq:Fleqsum} it must hold that $p_k^* = c_k^*$ and $\norm{\m{y}}_{\cA\sbra{\cI}} = \sum_{k} c_k^*$. Therefore, the FS atomic decomposition in \eqref{eq:yatomdec} must achieve the FS atomic norm.
\end{proof}

\begin{rem} Note that the SDP formulation of the FS atomic norm presented in Theorem \ref{thm:atomnD} can be easily extended to the multiple frequency band case by applying Corollary \ref{cor:BCVDM}, to be specific, by replacing the constraints in \eqref{eq:atomnDsdp} resulting from \eqref{eq:Tupsd} and \eqref{eq:Tunew} by those in \eqref{eq:sumTleqT}-\eqref{eq:Tunewl}. The proof of Theorem \ref{thm:atomnD} can still be applied in this case with minor modifications. \label{rem:SDP_MB}
\end{rem}

It immediately follows from Theorem \ref{thm:atomnD} that \eqref{eq:atomnmin_I} can be written as the following SDP:
\equ{\begin{split}
& \min_{\m{y}, x, \m{t}} \frac{1}{2}x + \frac{1}{2}t_0,\\
& \st \begin{bmatrix} x & \m{y}^H \\ \m{y} & \m{T} \end{bmatrix}\geq \m{0}, \m{T}_g\geq \m{0}  \text{ and } \m{y}_{\Omega} = \m{y}_{\Omega}^o. \end{split} \label{eq:CANM_sdp}}
Note that \eqref{eq:CANM_sdp} can be solved using off-the-shelf SDP solvers such as SDPT3. Given its solution, the frequencies can be retrieved from the FS Vandermonde decomposition of $\m{T}$. Moreover, as in the standard atomic norm method, the Toeplitz matrix $\m{T}$ in \eqref{eq:CANM_sdp} can be interpreted as the ``data covariance matrix'' \cite{yang2015gridless,yang2016vandermonde}. By solving \eqref{eq:CANM_sdp} we actually fit the data covariance matrix $\m{T}$ by exploiting its structures, e.g., PSDness (the first constraint), Toeplitz (explicitly imposed) and low rank ($t_0$ in the objective is proportional to the nuclear or trace norm of $\m{T}$), and its connection to the acquired data $\m{y}_{\Omega}^o$ (the first and the last constraints). But different from the standard atomic norm method, more precise knowledge of $\m{T}$ is exploited in the FS atomic norm method by additionally including the constraint $\m{T}_g\geq \m{0}$.

Before proceeding to the next subsection, we note that \eqref{eq:atomnmin_I} was solved by studying its dual in \cite{mishra2015spectral}. In particular, the dual of \eqref{eq:atomnmin_I} is given by:
\equ{\max_{\m{z}} \Re \sbra{\m{y}_{\Omega}^o \m{z}_{\Omega}}, \st \norm{\m{z}}_{\cA\sbra{\cI}}^*\leq 1 \text{ and } \m{z}_{\Omega^c}=\m{0}, \label{eq:dualprb}}
where $\Omega^c$ denotes the complement of $\Omega$ and $\norm{\m{z}}_{\cA\sbra{\cI}}^*$ is the dual FS atomic norm. By the fact that
\equ{\norm{\m{z}}_{\cA\sbra{\cI}}^* = \sup_{\m{a}\in\cA\sbra{\cI}} \Re\sbra{\m{a}^H\m{z}} = \sup_{f\in\cI} \abs{\m{a}^H\sbra{f}\m{z}}, }
the constraint that $\norm{\m{z}}_{\cA\sbra{\cI}}^*\leq 1$ can be cast as the following:
\equ{\abs{\m{a}^H\sbra{f}\m{z}}\leq 1 \text{ for any } f\in\cI, \label{eq:absazleq1}}
where
\equ{q(f)\coloneqq \m{a}^H\sbra{f}\m{z} \label{eq:qf}}
is referred to as the dual polynomial \cite{candes2013towards,mishra2015spectral}.
It follows that $1-\abs{q(f)}^2$ is a Hermitian trigonometric polynomial nonnegative on $\cI$ and, by the PRL, admits a Gram matrix parametrization as in \eqref{KPolM}. With some further derivations that we will omit, it can be shown that \eqref{eq:absazleq1} holds if and only if the unit polynomial (the right hand side of the inequality in \eqref{eq:absazleq1}) has the following Gram matrix parametrization:
\equ{\tr\sbra{\m{\Theta}_j \m{Q}_0} + \tr\mbra{ \m{\Theta}_{gj} \m{Q}_1} = \left\{ \begin{array}{ll} 1, & \text{if } j=0,\\ 0, & \text{otherwise}, \end{array}\right. \label{eq:sdpbrl1}}
where $\m{Q}_0$ and $\m{Q}_1$ satisfy
\equ{\begin{bmatrix}1 & \m{z}^H \\ \m{z} & \m{Q}_0 \end{bmatrix} \geq \m{0} \text{ and } \m{Q}_1\geq \m{0}. \label{eq:sdpbrl2}}
In fact, the characterization of \eqref{eq:absazleq1} using \eqref{eq:sdpbrl1} and \eqref{eq:sdpbrl2} is nothing but the result of the bounded real lemma (BRL) for trigonometric polynomials \cite{davidson2002linear,dumitrescu2007positive}. This can be viewed as a more precise result of the PRL when dealing with bounded polynomials as in \eqref{eq:absazleq1}. Finally, \eqref{eq:dualprb} is cast as the following SDP:
\equ{\max_{\m{z}, \m{Q}_0,\m{Q}_1} \Re \sbra{\m{y}_{\Omega}^o \m{z}_{\Omega}}, \st \eqref{eq:sdpbrl1}, \eqref{eq:sdpbrl2} \text{ and } \m{z}_{\Omega^c}=\m{0}. \label{eq:dualSDP}}
Without surprise, it follows from a standard Lagrangian analysis that \eqref{eq:dualSDP} is the dual of \eqref{eq:CANM_sdp} (note that the analysis uses \eqref{eq:TinTheta} and \eqref{eq:TginThetag} and will be left to interested readers). Since strong duality holds \cite{boyd2004convex}, the solution to \eqref{eq:dualSDP} can be obtained for free when solving \eqref{eq:CANM_sdp} using a primal-dual algorithm, and vice versa.

In summary, the FS Vandermonde decomposition can be applied to provide a primal SDP formulation of \eqref{eq:atomnmin_I}, while the trigonometric polynomial based technique in \cite{mishra2015spectral} provides a dual SDP formulation. Moreover, the FS Vandermonde decomposition also provides a new method for frequency retrieval. In fact, it is found that the new method results in higher numerical stability, as compared to the root-finding method in \cite{candes2013towards,mishra2015spectral}. This can be explained as follows. By using the FS Vandermonde decomposition, we can always determine the number of frequencies first by computing $\rank\sbra{\m{T}}$, which can effectively reduce the problem dimension and improve stability. In contrast to this, the root-finding method requires to solve all, up to $2N-2$, roots of the polynomial $1-\abs{q(f)}^2$, among which appropriate ones (those with unit modulus) are then selected to produce the frequencies.

\subsection{Computational Complexity}
We next analyze the computational complexity of the presented FS atomic norm method, to be specific, the complexity of solving the SDP in \eqref{eq:CANM_sdp}. To do so, we consider the general multiple band case in which, according to Remark \ref{rem:SDP_MB}, \eqref{eq:CANM_sdp} becomes:
\equ{\begin{split}
& \min_{\m{y}, x, \m{t}_l} \frac{1}{2}x + \frac{1}{2}\sum_{l=1}^J t_{l0},\\
& \st \begin{bmatrix} x & \m{y}^H \\ \m{y} & \sum_{l=1}^J \m{T}(\m{t}_l) \end{bmatrix}\geq \m{0}, \\
&\phantom{\st} \m{T}(\m{t}_l) \geq \m{0}, \; \m{T}_g(\m{t}_l)\geq \m{0}, l=1,\dots,J,\\
&\phantom{\st} \m{y}_{\Omega} = \m{y}_{\Omega}^o. \end{split} \label{eq:CANM_sdp_MB}}
Evidently, the SDP in \eqref{eq:CANM_sdp_MB} has $n=O(JN)$ free variables and $m=2J+1$ LMIs, and the $i$th LMI has size of $k_i\times k_i$ with $k_i=O(N)$. It follows from \cite{ben2013lectures} that a primal-dual algorithm for \eqref{eq:CANM_sdp_MB} has a computational complexity on the order of
\equ{\sbra{1+\sum_{i=1}^m k_i}^{\frac{1}{2}} n\sbra{n^2+n\sum_{i=1}^m k_i^2 + \sum_{i=1}^m k_i^3} = O\sbra{J^{3.5}N^{4.5}}. \label{eq:complexity}}
By arguments similar to those above, the standard atomic norm method in the absence of prior knowledge has a computational complexity of $O\sbra{N^{4.5}}$. This together with
\eqref{eq:complexity} indicates that, with a fixed number of intervals $J$, the presented FS atomic norm method has a complexity higher than the standard atomic norm method by a constant factor and the factor increases with $J$.

\subsection{Numerical Simulation}
We provide a simple illustrative example below to demonstrate the advantage of using the prior knowledge for frequency estimation.

\begin{exa} Consider a line spectrum composed of $K=3$ frequencies $\m{f}=[0.22, 0.23, 0.28]^T$ as shown in Fig. \ref{Fig:LSE}. To estimate/recover the spectrum, $M=16$ randomly located noiseless samples are acquired among $N=64$ uniform samples. The standard ANM and the FS-ANM methods are implemented using SDPT3 to estimate the line spectrum. In FS-ANM, the prior knowledge that the frequencies lie in $\cI=\mbra{0.2, 0.3}$ is used. The estimation results are presented in Fig. \ref{Fig:LSE}. It can be seen that FS-ANM exactly recovers the spectrum but ANM does not. For both ANM and FS-ANM, the recovered frequencies retrieved using the Vandermonde decomposition match the locations at which the dual polynomials have unit magnitude. For FS-ANM the frequencies computed using the FS Vandermonde decomposition have recovery errors on the order of $10^{-10}$ while those computed using the root-finding method have errors on the order of $10^{-6}$.
\end{exa}

\begin{figure}
\centering
  \subfigure[ANM]{
    \label{Fig:ANM}
    \includegraphics[width=3in]{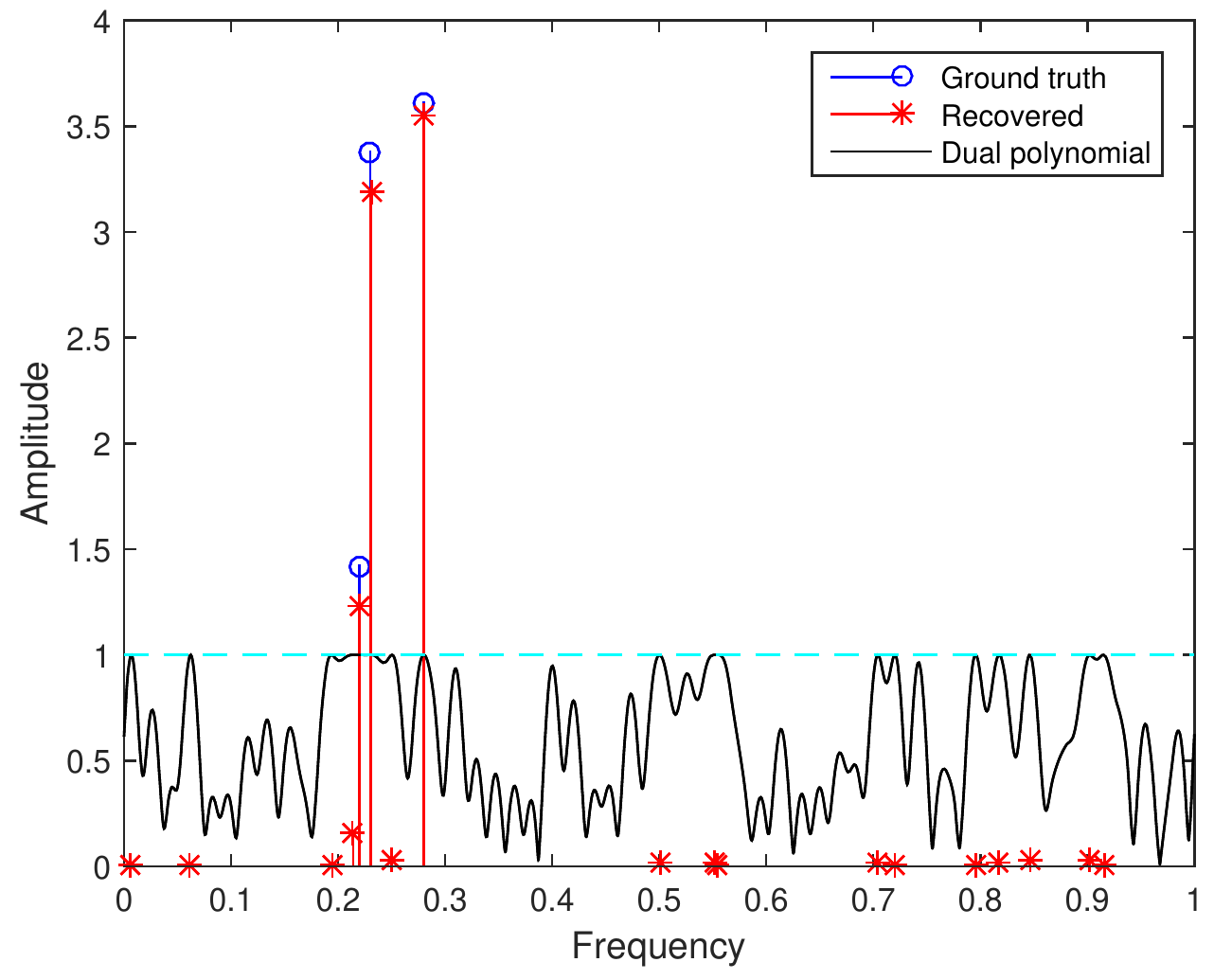}}%
  \subfigure[FS-ANM]{
    \label{Fig:CANM}
    \includegraphics[width=3in]{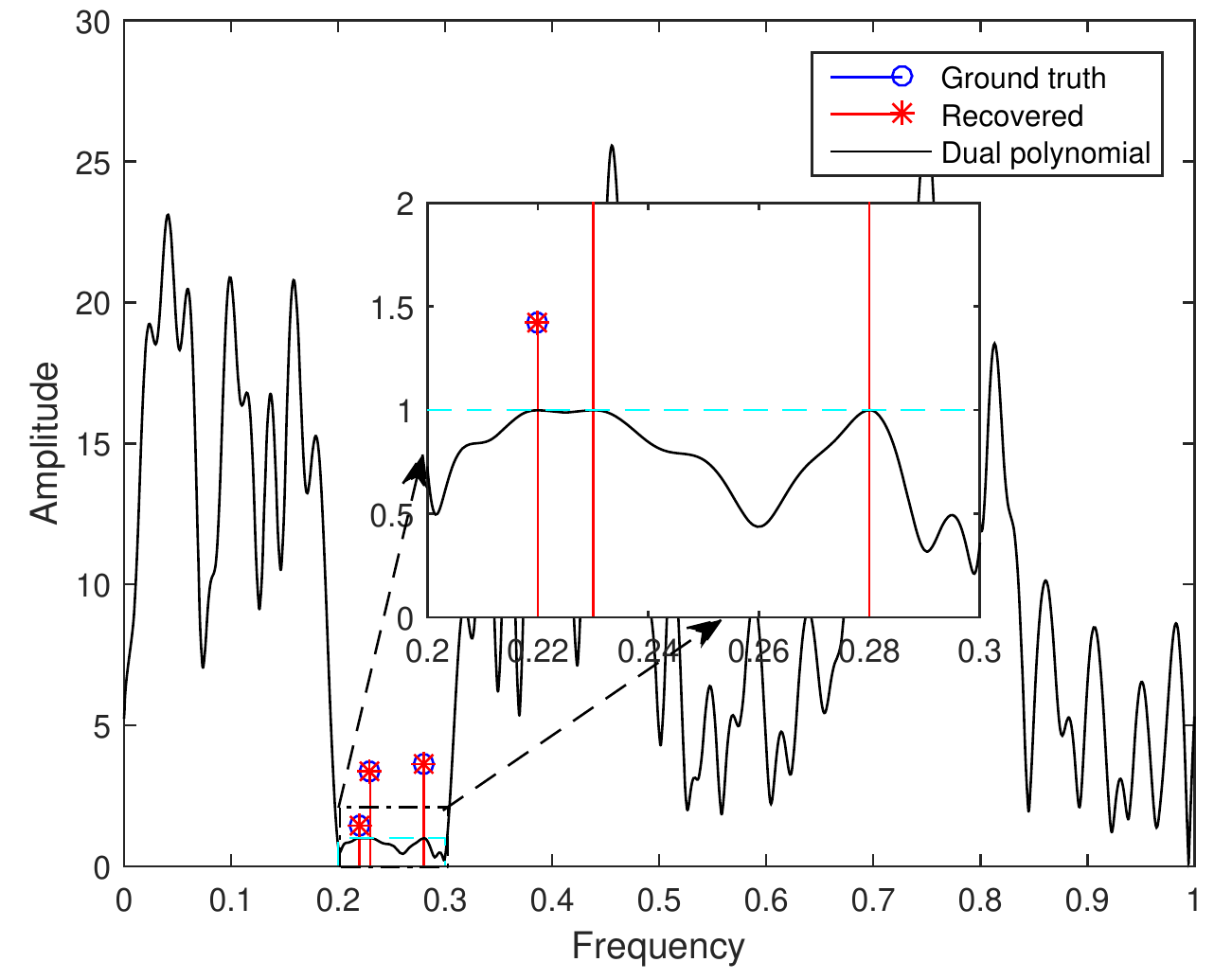}}
\centering
\caption{Line spectral estimation results of (a) ANM and (b) FS-ANM.} \label{Fig:LSE}
\end{figure}

Note that the presented method can deal with noise with minor modifications, as shown in \cite{mishra2015spectral}. In the noisy case, a simulation has been included in \cite{mishra2015spectral} to compare the signal recovery errors of the atomic norm method in cases with and without the prior knowledge. It is shown that ``the prior information formulation yields a higher stability in presence of noise.'' Readers are referred to \cite[Section VIII-B]{mishra2015spectral} for detail.


\subsection{Extension to FS Atomic $\ell_0$ Norm}
In this subsection, we provide an example in which the FS Vandermonde decomposition result is applicable but the theory of trigonometric polynomials is not. In particular, we study the FS atomic $\ell_0$ norm defined by:
\equ{\begin{split}\norm{\m{y}}_{\cA\sbra{\cI},0} \coloneqq
&\inf_{c_k>0,\m{a}_k\in\cA\sbra{\cI}}\lbra{\cK:\; \m{y} = \sum_{k=1}^{\cK} c_k \m{a}_k}\\ =
&\inf_{f_k\in\cI,s_k}\lbra{\cK:\; \m{y} = \sum_{k=1}^{\cK} \m{a}\sbra{f_k}s_k}. \end{split} \label{eq:atom0n_I}}
$\norm{\m{y}}_{\cA\sbra{\cI},0}$ is of interest since it exploits sparsity to the greatest extent possible, while $\norm{\m{y}}_{\cA\sbra{\cI}}$ is in fact its convex relaxation. It has been vastly demonstrated in the literature on compressed sensing that improved performance can usually be obtained by solving (or approximately solving) $\ell_0$ norm based problems (see, e.g., \cite{candes2008enhancing,andersson2014new,yang2016vandermonde}). More recently, a new trend of frequency estimation is to directly solve the $\ell_0$ norm based formulations using nonconvex optimization techniques for low rank matrix recovery \cite{cho2016fast,cai2016fast}. To do so, the key is to formulate the frequency estimation problem in the continuous setting as a matrix rank minimization problem. In the context of the FS atomic $\ell_0$ norm, the following result can be obtained by applying the FS Vandermonde decomposition.

\begin{thm} It holds that
\equ{\begin{split} \norm{\m{y}}_{\cA\sbra{\cI},0}
=& \min_{x, \m{t}} \rank\sbra{\m{T}},\\
& \st \begin{bmatrix} x & \m{y}^H \\ \m{y} & \m{T} \end{bmatrix}\geq \m{0} \text{ and } \m{T}_g\geq \m{0}, \end{split} \label{eq:atom0nDsdp}}
where $g$ is as defined previously.
 \label{thm:atom0nD}
\end{thm}

\begin{proof} The proof is similar to that of Theorem \ref{thm:atomnD}. At the first step, by applying the FS Vandermonde decomposition, we can construct a feasible solution, as in the proof of Theorem \ref{thm:atomnD}, to the optimization problem in \eqref{eq:atom0nDsdp}, which concludes that $\norm{\m{y}}_{\cA\sbra{\cI},0}\leq r^*$, where $r^*$ denotes the optimal objective value of \eqref{eq:atom0nDsdp}. At the second step, for any optimal solution that achieves the optimal value $r^*$, we can similarly obtain an $r^*$-atomic FS decomposition of $\m{y}$, which results in that $\norm{\m{y}}_{\cA\sbra{\cI},0}\geq r^*$. So we complete the proof.
\end{proof}

It follows from Theorem \ref{thm:atom0nD} that $\norm{\m{y}}_{\cA\sbra{\cI},0}$ can be cast as a rank minimization problem, while solving (or approximately solving) the resulting optimization problem is beyond the scope of this paper. It is worth noting that, since $\norm{\m{y}}_{\cA\sbra{\cI},0}$ is nonconvex, a trigonometric polynomial based technique, as used for $\norm{\m{y}}_{\cA\sbra{\cI}}$ in \cite{mishra2015spectral}, cannot be applied in this case to provide a finite-dimensional formulation.

\section{Conclusion} \label{sec:conclusion}

In this paper, the FS Vandermonde decomposition of Toeplitz matrices on a given interval was studied. The new result generalizes the classical Vandermonde decomposition result. It was shown by duality to be connected to the theory of trigonometric polynomials. It was also applied to provide a solution to the classical truncated trigonometric $K$-moment problem and a primal SDP formulation of the recent FS atomic norm for line spectral estimation with prior knowledge.


\end{document}